\documentclass[runningheads,a4paper]{llncs}
\pdfoutput=1

\usepackage[latin1]{inputenc}
\usepackage{fontenc}

\usepackage{wrapfig}
\usepackage[plain,vlined,linesnumbered]{algorithm2e}   
\usepackage[matrix,arrow]{xy}

\usepackage{color}

\usepackage{bigstrut}
\usepackage[pdftex]{graphicx}

\newcommand{\myitem}{\smallskip\item[$\bullet$]}

\newcommand{\myparagraph}[1]{\medskip \noindent {\textbf{#1}}}  

\newcommand{\ACKNOWLEDGMENT}{\medskip \noindent {\textbf{Acknowledgement:}}}

\newcommand{\mycomment}[1]{}  

\newcommand{\mytrue}{\text{\bf 1}}                    
\newcommand{\myfalse}{\text{\bf 0}}                   

\usepackage{amsmath}
\usepackage{amsfonts}
\usepackage{amssymb}

\usepackage{fancyvrb}
\usepackage{xcolor}

\xdefinecolor{monvert}{rgb}{0,.5, 0}
\xdefinecolor{lsvblue}{rgb}{0, .4, .6}
\xdefinecolor{LemonChiffon}{rgb}{1,.98,.8}
\xdefinecolor{Vert}{rgb}{0.07,0.5,0.4}
\xdefinecolor{bleu}{rgb}{0,.2,.54}
\xdefinecolor{Ver}{rgb}{0.0525,0.375,0.3}
\xdefinecolor{Rouge}{rgb}{1,0,0}
\xdefinecolor{Bleu}{rgb}{0,0,1}
\xdefinecolor{seborange}{rgb}{1,0.75,0}

\newcommand{\rouge}[1]{\textcolor{red}{#1}}

\newcommand{\monvert}[1]{\textcolor{monvert}{#1}}

\newcommand{\mycheckmark}{\monvert{\bf \Large  \checkmark}}
\newcommand{\mybadmark}{\rouge{\bf \Large  $\times$}}

\newcommand{\TE}{\mbox{\bf $T_E$}}
\newcommand{\TUF}{\mbox{\bf $T_{UF}$}}
\newcommand{\TA}{\mbox{\bf $T_{A}$}} 
\newcommand{\SE}{\Sigma_{E}}  
\newcommand{\SUF}{\Sigma_{UF}}   
\newcommand{\SA}{\Sigma_{A}}  
\newcommand{\C}{\mbox{\bf $C$}}  

\newcommand{\mydef}{\triangleq}

\newcommand{\N}{\mathbb{N}}
\newcommand{\B}{\mathbb{B}}

\newcommand{\CCA}{{\sc cc}}
\newcommand{\CPFD}{{\sc fd}}
\newcommand{\fdcc}{{\sc fdcc}}

\newcommand{\formula}[1]{\phi \langle #1 \rangle}

\newcommand{\csp}{\mathcal{R}}
\newcommand{\Vars}{\mathcal{X}}
\newcommand{\Doms}{\mathcal{D}}
\newcommand{\Cstrs}{\mathcal{C}}
\newcommand{\Legal}[1]{L_{#1}}
\newcommand{\legal}[2]{L_{#2}(#1)}   
\newcommand{\DDomain}{\mathcal{U}}
\newcommand{\var}{x}
\newcommand{\dom}{D}
\newcommand{\cstr}{c}

\begin{document}

 \title{A Combined Approach for \\ Constraints over  Finite Domains and Arrays\thanks{A preliminary version of this paper was presented at  CPAIOR 2012~\cite{BG-12}}}

\author{S\'ebastien Bardin\inst{1} \and Arnaud Gotlieb\inst{2}} 

\institute{
CEA, LIST, Gif-sur-Yvette, F-91191, France \\ \email{sebastien.bardin@cea.fr} 
\and 
Certus Software V\&V Center, SIMULA RESEARCH LAB., Lysaker, Norway \\ \email{arnaud@simula.no}
}

\maketitle

\begin{abstract}
Arrays are ubiquitous in the context of software verification.
 However, effective reasoning over arrays is still rare in CP, as local reasoning is dramatically ill-conditioned for 
 constraints over arrays.
 In this paper, we propose an approach combining both global symbolic reasoning and local consistency filtering in order to solve constraint systems 
involving  arrays (with accesses, updates and size constraints) 
  and finite-domain constraints over their elements and indexes.  
Our approach, named \fdcc,    is based  
 on a combination of a  congruence closure algorithm for the standard theory of arrays and   a  CP solver over finite domains.
The tricky part of the work lies in the
 bi-directional  communication mechanism between both solvers. We identify the significant information to share, 
and design ways to master the communication overhead. 
 Experiments on random instances show that \fdcc\    
 solves more formulas than any portfolio combination of the two solvers taken in isolation,    
while overhead is kept  reasonable. 
\end{abstract}

\section{Introduction}


\myparagraph{Context.} Constraint resolution is  an emerging trend in software verification \cite{Rus08}, 
either to  automatically generate test inputs or formally prove
some properties of a program.  
Program analysis  involves solving  so-called  Verification Conditions, i.e.~checking the satisfiability of a formula either 
by providing a solution ($sat$) or showing there is none ($unsat$). 
While most techniques are based on SMT (Satisfiability Modulo Theory),  
a few  verification tools~\cite{BH08,BH11,CRH08,GBR00,Gotlieb-09,MB05} rely on  Constraint Programming over Finite Domains, denoted CP(FD).     
%
%
%
CP(FD) is appealing here because it allows to reason about some fundamental aspects of programs notoriously difficult to handle, 
like floating-point numbers \cite{BCG13}, 
bounded non-linear integer arithmetic, 
modular arithmetic \cite{GLM10,GSS13} or bitvectors \cite{BHP10}.    
Some experimental evaluations \cite{BHP10,CRH08} suggest that CP(FD) can be an interesting 
alternative to SMT for certain classes of Verification Conditions.


\myparagraph{The problem.} 
Yet the effective use of CP(FD) in program verification is limited by 
the absence of effective methods to handle complex constraints over arrays.    
 Arrays are non-recursive data structures that can be found in most programming languages and thus, 
 checking the satisfiability of formulas involving arrays is of primary importance in program verification. 
Moreover, resolution techniques for constraints involving arrays can often be leverage to constraints 
over  data types like maps \cite{BM-07} and memory heaps \cite{Bornat-00}.

%
While array accesses are handled for a long time through the {\sc Element} constraint  \cite{HC88},  
array updates have been dealt with only recently \cite{CBG09}, and in both cases the reasoning relies only on 
 local consistency filtering. This is insufficient to handle constraints involving long chains of accesses and updates arising in program verification.

On the other hand, the theory of array is well-known in theorem proving \cite{BM-07,KS08}. 
Yet, this theory  cannot express size constraints over arrays nor   
 domain constraints over elements and indexes. 
A standard solution is to consider  a combination of two decision procedures, one for the array part  and one for the index and element part, through a standard 
cooperation framework like the Nelson-Oppen (NO)  scheme  \cite{NO79}. 
Indeed, under some theoretical conditions, NO provides a mean to build
a decision procedure for a combined theory $T \uplus T'$ from existing decision procedures for $T$ and $T'$\mycomment{,  through term equality propagation}. Unfortunately, finite-domain constraints cannot 
be integrated into NO since eligible theories must have an infinite model \cite{NO79}.

\myparagraph{Contributions.} This paper addresses the problem of designing an efficient CP(FD) approach for solving conjunctive quantifier-free formulas 
combining fixed-size arrays  and finite-domain constraints over indexes and elements.   
Our main guidelines are (1) to combine global symbolic deduction mechanisms  with local consistency filtering in order to achieve better 
deductive power than both technique taken in isolation,  
(2) to keep communication overhead as low as possible, while going beyond a purely portfolio combination of the two approaches, 
(3) to design a   combination scheme allowing to re-use any existing FD solver in a black box manner, with minimal and easy-to-implement API.     
Our main contributions are the following:
\begin{itemize}
\myitem We design \fdcc,  an original decision procedure built upon a 
 lightweight congruence closure algorithm for the theory of arrays, called \CCA\  in the paper,   
interacting with 
a local consistency filtering CP(FD) solver, called \CPFD.  
 To the best of our knowledge, it is the first collaboration scheme including a finite-domain CP solver 
and a Congruence Closure solver  for  array constraint systems. 
Moreover, the combination scheme, while more intrusive than   NO,  is still high-level. 
Especially, \CPFD\  can be used in a black-box manner through a minimal API,   and large parts of  \CCA\  are standard. 
%


 \myitem We bring new ideas to make both solvers cooperate through  bi-directional constraint
exchanges and synchronisations. We identify  important classes of information to be exchanged, and propose  ways of doing it efficiently :    
on the one side, the congruence closure algorithm can send equalities, disequalities and {\sc Alldifferent} constraints to \CPFD,  
 while  on the other side, \CPFD\  can deduce new  equalities / disequalities from local consistency filtering and send them to   \CCA. 
In order to master the communication overhead, a \textit{supervisor} queries explicitly the most expensive computations, 
while cheaper deductions are propagated asynchronously.   

%
%
%


\myitem We propose an implementation of our  approach written on  top of SICStus {\tt clpfd}. 
Through experimental results on 
random instances, we show that \fdcc\  systematically solves more formulas 
that \CCA\  and \CPFD\  taken in isolation. \fdcc\  performs even better 
than the best possible portfolio combination of the two  solvers.   Moreover, \fdcc\  shows only a  reasonable overhead 
over \CCA\  and \CPFD. 
This is particularly interesting in a verification  setting, since it means that \fdcc\ can be clearly 
preferred to the standard \CPFD-handling of arrays in any context, i.e.~whether we want to solve a few complex formulas 
or we want to solve as many as formula in a short amount of time.

\myitem We discuss how the \fdcc\ framework can handle  other array-like structures of interest in software verification, namely uniform arrays, arrays with non-fixed (but bounded)  size  and 
maps\mycomment{(with addition and deletion of keys)}. Noticeably, this can be achieved without any change to the framework,  by considering  only  
  extensions of the \CCA\ and \CPFD\ solvers.

\end{itemize}
 

%

 \myparagraph{Outline.} The rest of the paper is organised as follows.  Section \ref{sec:motivating} introduces running examples  used throughout the paper. Section \ref{sec:background} presents 
a few  preliminaries  while Section \ref{sec:arrayctr} describes the theory of arrays and its standard  decision procedures.     
Section \ref{sec:fdcc} describes our technique to combine congruence
 closure with a finite domain constraint solver. Section \ref{sec:implem-experiments} presents our implementation \fdcc\  and  experimental results. 
 Section \ref{sec:extensions} describes extensions to richer array-like structures.    Section \ref{sec:related} discusses related work.  
 Finally, Section \ref{sec:conclusion} concludes the paper.
\medskip

\section{Motivating examples} \label{sec:motivating}




\begin{figure}[ht]
\hrulefill \smallskip 

\centering
\begin{minipage}{.05\linewidth}
\small
\begin{tabbing}
tt \= tt \= tttttttttttttttttttttttttttttttttttt \= tt  \kill
    \> {\bf Prog1}                        \>\>  {\bf Prog2}     \\
    \> {\tt int A[100];  \ldots}                   \>\> {\tt int A[2];  \ldots}     \\
    \> {\tt int e=A[i]; int f=A[j];}         \>\> {\tt int e=A[i]; int f=A[j]; int g=A[k];} \\ 
    \> {\tt if (e != f \&\& i = j) \{ \ldots}  \>\>  {\tt if (e != f \&\& e != g \&\& f != g) \{ \ldots}
\end{tabbing}
 \end{minipage}

\smallskip \hrulefill 

 \caption{Programs with arrays}\label{ex1}
\end{figure}



We use the two programs of  Fig.~\ref{ex1} as running examples. 
%
First, consider the problem of generating a test input
satisfying the decision in program Prog1 of Fig.~\ref{ex1}.  
This involves solving a constraint system with array accesses, namely
\begin{equation}
\mbox{\sc element}(i,A,e),\mbox{\sc element}(j,A,f),e \neq f,i=j
\end{equation}

\noindent where $A$  is  an array of variables of size $100$,  and $\mbox{\sc element}(i,A,e)$ means $A[i] = e$.  
A  model of this constraint system written in OPL for CP Optimizer \cite{OPL}  
did not
provide us with an $unsat$ answer within $60$ minutes of CPU time on a standard machine. 
In fact, as only local consistencies are used in the underlying solver, the system cannot infer 
that $i \neq j$ is implied by the three first constraints.  On the contrary, a  SMT solver such as Z3 \cite{Z3}   immediately 
reports $unsat$, using a global   symbolic decision procedure for the standard theory of arrays.   

\medskip

Second, consider the problem of producing a test input
satisfying the decision in program Prog2 of Fig.~\ref{ex1}. 
It requires solving the following constraint system: 
\begin{equation}
\mbox{\sc element}(i,A,e),\mbox{\sc element}(j,A,f),\mbox{\sc element}(k,A,g),
e \neq f,e \neq g, f \neq g
\end{equation}
where $A$ is an array of size $2$. 
 A  symbolic decision procedure considering only the standard theory of arrays
returns  (wrongly) a $sat$ answer here  while the formula is unsatisfiable,  since    
$A[i],A[j]$ and $A[k]$ cannot take three distinct values. 
%
%
%
To address the problem, a symbolic approach for arrays must be combined with  an explicit encoding of all possible values of indexes.    
However, this encoding is expensive, requiring to add many disjunctions (through enumeration). 
%
%
%
On this example, a CP solver over finite domains can also fail to return $unsat$ in a reasonable amount of time if
it starts labelling on elements instead of indexes, as  nothing
prevents  to consider constraint stores where $i=j$ or $i=k$ or $j=k$:  there is no 
  \textit{global} reasoning over  arrays able to deduce from $A[i] \neq A[j]$ that $i \neq j$. 
%
%


\section{Background}  \label{sec:background}

We describe hereafter a few theories closely related to the theory of arrays, the standard congruence closure algorithm  
and basis of constraint programming. We also recall a few facts about decision procedure combination. 
If not otherwise stated, we consider only conjunctive fragments of quantifier-free  theories.

\subsection{Theory of equality and theory of uninterpreted functions} \label{sec:E-UF}

A logical theory is a first-order language with a restricted set of permitted functions and predicates, together with their axiomatizations. 
We present here two standard theories closely related to the theory of arrays (presented in Section \ref{sec:arrays}): the theory of equality \TE\ and the theory of uninterpreted functions \TUF.  

\begin{itemize}
\myitem \TE\  has signature $\SE \mydef \{=, \neq\}$, i.e.,~the only available predicate is (dis-)equality and no function symbol is allowed. 
Formulas in \TE\ are of the form  $x=y \wedge a \neq b$, where variables are {\it uninterpreted}  in the sense that they do not range over any implicit domain.

\myitem \TUF\  extends $\TE$ with signature $\SUF \mydef \{=, \neq, \{f_1, f_2, f_3,\ldots\}\}$ where the $f_i$ are function symbols. 
Formulas in \TUF\ are of the form  $f(x,y) = g(a) \wedge a \neq h(b)$. 
  Variables and functions are  uninterpreted, i.e., the only assumption about any $f_i$ is  its functional consistency  (FC):   $\forall f_i. \forall x,y. x = y \Longrightarrow f_i(x) = f_i(y)$     
\footnote{ \TUF\ does not assume a free-algebra of terms (as Prolog does), allowing for example to find solutions for constraint $f(a) = g(b) = 3$. \TUF\ can be extended with a free-algebra assumption of the form $\forall f_1, f_2. \forall x, y. f_1(x) \neq f_2(y)$.}.

\end{itemize}

While not very expressive,  \TE\ and \TUF\  enjoy polynomial-time satisfiability checking. Standard decision procedures are based on Congruence closure (Section \ref{sec:congruence-closure}). 
Note that allowing disjunctions makes the satisfiability problem NP-complete.

\myparagraph{Interpreting variables.}  
While variables are uninterpreted, it is straightforward to encode a set of constant values $k_1,\ldots,k_n$ through introducing  new variables $x_{k_1}, \ldots,x_{k_n}$  
together with the corresponding disequalities between the $x_{k_i}$'s (e.g., $x_{k_i} \neq x_{k_j}$ if $k_i = 2$ and $k_j = 3$). 
Adding  domains to variables is more involving. Finite-domain constraints can be explicitly encoded with disjunctions ($x \in D$ translates into $\vee_{k \in D}\  x = k$), but the underlying satisfaction problem becomes NP-complete. For variables defined over an arbitrary theory $T$, one has to consider the combined theory $\TUF \uplus T$. The DPLL($T$) framework and the Nelson-Oppen combination scheme 
can be used to recover decision procedures  from available decision procedures over \TE, \TUF\ and $T$ (see Section \ref{sec:no-dpll}).

\subsection{The congruence closure algorithm}  \label{sec:congruence-closure} 


The congruence closure algorithm aims at computing all equivalence classes of a relation over a set of terms~\cite{NO80}. 
It also provides a decision procedure for the theory \TE. 
The algorithm relies on a {\it union-find structure}
to represent the set of
all equivalence classes.   
Basically, each class of equivalence has a unique witness 
and each term is (indirectly) linked  to its witness.  Adding an equality between two terms amounts 
to  choose one term's witness to be the witness of the other term. 
Disequalities inside the same equivalence class lead to $unsat$. 
%
Smart handling of ``witness chains'' through \textit{ranking} and \textit{path compression} ensures  very efficient implementations in $O(n)$.  
%
%
%
%
We sketch such an algorithm in Fig.~\ref{figCC}. Each initial variable $x$ is associated with two fields: $parent$ and $rank$. Initially, $x.parent=x$ and $x.rank=0$. 
Path compression is visible at line 3 of 
the {\it find} procedure. Ranking optimisation amounts to compute the rank of each variable, and choose the variable with larger rank as the new witness in $union$.   

\begin{figure}[htb]

\begin{center}
\fbox{\small
\begin{minipage}[t]{7cm} 
\quad \begin{algorithm}[H]
{{\bf function}  $union(x,y)$: } \\




     $x'$ := $find(x)$ \; 
     $y'$ := $find(y)$ \; 

     \lIf{$x'$ == $y'$}{ skip \;}
     \uElseIf{ $x'.rank < y'.rank$}{
         $x'.parent$ := $y'$\; 
         }
     \uElseIf{$x'.rank  > y'.rank$}{
         $y'.parent$ := $x'$\; 
         }
     \Else{
         $y'.parent$ := $x'$ \; 
         $x'.rank$ := $x'.rank$ + 1 \; 
          }
      \Return\; 
\end{algorithm}
\end{minipage} \ 
\hfill 
\begin{minipage}[t]{6.5cm}
\begin{minipage}{6.5cm}
\begin{algorithm}[H]
{\bf function} $find(x)$:\\
 \If{$x.parent \neq x$}{$x.parent$ := $find(x.parent)$ } 
 \Return($x.parent$)\;
\end{algorithm}
\end{minipage}

\medskip

\begin{minipage}{6.5cm}
\begin{algorithm}[H]
{\bf function} $create(x)$:\\
           $x.parent$ := $x$ \; 
         $x.rank$ := 0 \; 
\end{algorithm}
\end{minipage}

\end{minipage}
}

  \caption{Congruence closure algorithm}\label{figCC}
\end{center}
\end{figure}


The algorithm presented so far works for \TE\ and can be  extended to \TUF\ with only slight modification taking sub-terms into account~\cite{NO80}. 
The procedure remains   polynomial-time.  

\subsection{Combining solvers}  \label{sec:no-dpll}

The { Nelson-Oppen cooperation scheme} (NO) allows to combine  two solvers $S_{T} : T \mapsto \B$ and $S_{T'} : T' \mapsto \B$ for theories $T$ and $T'$ 
into a solver for the combined theory $T \uplus T'$. 
%
Theories $T$ and $T'$ are essentially required  \cite{NO79} to be disjoint (they may share only the $=$ and $\neq$ predicates) and stably-infinite (whenever  a model of a formula exists, an infinite model must exist as well). Suitable theories include \TE, \TUF, the theory of arrays and the theory of linear (integer) arithmetic. However, {\it finite-domain constraints do not satisfy these assumptions}. 
 Moreover, in the case of non-convex theories (including arrays and linear integer arithmetic), theory solvers must be able to propagate all 
\textit{implied disjunctions of equalities}, which is  harder than satisfiability checking \cite{BCFGS-09}.

The { DPLL($T$) framework} \cite{Tinelli-2002} takes advantage of a DPLL  SAT-solver  in order to leverage  a solver $S_{T} : T \mapsto \B$   
into a solver for $T_{\wedge,\vee}$.  
Propagation of implied disjunctions of equalities in NO 
is   reduced to the propagation of implied equalities  at the price of letting DPLL  decides (and potentially backtracks) over all  possible equalities between 
variables.

\subsection{Contraint Programming over Finite Domains}\label{sec:cp}


Constraint Programming over Finite Domains, denoted CP(FD), deals with solving  satisfiability or optimisation problems for constraints defined over finite-domain variables. 
Standard CP(FD) solvers interleave two processes for solving constraints over finite domain variables, namely {\it local consistency filtering} and {\it labelling search}. 
Filtering narrows the domains of possible values of variables, 
removing some of the values which do not participate in any solution. 
When no more filtering is possible, search and backtrack take place.   
These procedures can be seen as generalisations of the DPLL procedure. 


Let $\DDomain$ be a {\it finite} set of values.
A constraint satisfaction problem (CSP) over $\DDomain$   
is a triplet $ \csp = \langle \Vars, \Doms, \Cstrs \rangle$
where 
 the domain $\Doms \subseteq \DDomain$  is a finite Cartesian product $\Doms=\dom_1\times\ldots\times\dom_n$, 
$\Vars$ is a finite set of variables $\var_1,\ldots,\var_n$ such that each variable $\var_i$ ranges over 
  $\dom_i$  
and $\Cstrs$ is a finite set of constraints $\cstr_1,\ldots, \cstr_m$ such that each constraint $\cstr_i$ is associated 
with a set of solutions $\Legal{\cstr_i} \subseteq \DDomain$. 
%
%
The set $\Legal{\csp}$ of solutions of $\csp$ is equal to $\Doms \cap \bigcap_{i}{\Legal{\cstr_i}}$. 
%
%
%
%
A value of $\var_i$ participating in a solution of $\csp$ is called a legal value, otherwise it is said to be spurious. In other words, the set $\legal{\var_i}{\csp}$ of legal values of $\var_i$ in $\csp$ is defined as the i-th projection of $\Legal{\csp}$.  
%
%
%
%
%
A propagator $P$ refines a CSP $\csp = \langle \Vars, \Doms, \Cstrs \rangle$ into another CSP $\csp' = \langle \Vars, \Doms', \Cstrs \rangle$ with $\Doms'  \subseteq \Doms$. 
%
A propagator $P$ is correct (or ensures correct propagation) if 
$ \legal{\var_1}{\csp}\times\ldots\times\legal{\var_n}{\csp} \subseteq \Doms' \subseteq \Doms$. The use of correct propagators ensures that no legal value is lost during propagation, which in turn ensures that no solution is lost, i.e.~$\Legal{\csp'}=\Legal{\csp}$. 
%
%

%
{ Local consistency filtering} considers each constraint individually to filter the domain of each of its variables. Several local consistency properties can be used, but the most common are domain-- and  bound--consistency \cite{COC97}. 
Such propagators  are considered as an interesting trade-off between  large pruning and fast propagation.

\section{Array constraints}\label{sec:arrayctr}

We present now the (pure) theory of arrays $\TA$ \-- no domain nor size constraints, 
two standard symbolic procedures  
for deciding the satisfiability of $\TA$-formulas and how CP(FD) can be used to handle a variation of $\TA$, adding finite domains to indexes and elements  
while fixing array sizes.

\subsection{The theory of arrays}  \label{sec:arrays}

%
The theory of arrays $\TA$ has signature $\SA =  \{select, store, =, \neq\}$,  where  $select(A,i)$ returns the value
of array $A$ at index $i$ and  $store(A, i, e)$ returns the array obtained from $A$ by  putting element $e$ at index $i$,
all other elements remaining unchanged.  
$\TA$ is typically described using the  {\it read-over-write semantics}~\cite{BM-07,KS08}.   
Besides the standard axioms of equality, three axioms dedicated to $select$ and $store$ are considered (cf.~Figure~\ref{array:axioms}).  
Axiom FC is an instance of the classical \textit{functional consistency} axiom, while RoW-1 and RoW-2 are two
variations of the \textit{read-over-write} principle (RoW).  


\begin{figure}[htp]
\centering
\fbox{\small

\begin{tabular}{lcl}

\textit{(FC) \qquad }  & \ \bigstrut   &  $   i = j \Longrightarrow select(A, i) = select(A, j)  $ \\


\textit{(RoW-1) \qquad }  & \bigstrut  &  $ i = j \Longrightarrow select(store(A, i, e), j) = e  $ \\

\textit{(RoW-2) \qquad }  & \bigstrut  &  $ i \neq j \Longrightarrow select(store(A, i, e), j) = select(A, j)  $ \\



\end{tabular}
}
\caption{Axioms for the theory of array $\TA$ } \label{array:axioms}
\end{figure}


Note that $\TA$  by itself does not express anything about the size of arrays, and both  
indexes and elements are uninterpreted (no implicit domain). 
Moreover, the theory is {\it non-extensional}, meaning that it cannot reason  on arrays themselves. For example,  $A[i] \neq B[j]$ is permitted, while $A \neq B$ and $store(A,i,e) = store(B,j,v) $   are not. 
Yet, array formula are  difficult to solve: 
the satisfiability problem for   the {\it conjunctive  fragment}   is  already NP-complete \cite{DS78}.

\myparagraph{Modelling program semantics.}  We give here a small taste of how $\TA$ can be used 
to model the behaviour of programs with  arrays. 
More details \mycomment{on modelling program semantics with logical theories in a verification setting} can be found in the literature~\cite{Bornat-00}.   
There are two main differences between arrays found in imperative programming languages such as C and the ``logical'' arrays defined in $\TA$.  
First,  logical arrays have no size constraints while real-life arrays have a fixed size. The standard solution here is to combine $\TA$ with  arithmetic  
 constraints expressing that each $select$ or $store$ index must be smaller than the  size of the array, arrays being coupled to a variable representing their size.  
Second,  real-life arrays can be accessed beyond their bounds, leading to typical bugs. Such buggy accesses are usually not directly taken into account in the formal modelling    in order to avoid 
the subtleties of reasoning over undefined values. The preferred approach is to add extra verification conditions asserting that all array accesses are
valid, 
and to verify separately  the program specifications (assuming all array accesses are within bounds).
%


\subsection{Symbolic algorithms for the theory of arrays}  \label{sec:array-symbolic}

 Symbolic decision procedures for $\TA$  
rely on the congruence closure algorithm shown above.  
There are two main classes of procedures \cite{BM-07,KS08}: 
\begin{itemize}
\myitem Create a dedicated \TA-solver through extending the congruence closure algorithm  with rewriting rules inspired from the array axioms. Case-splits are required for dealing with 
the RoW axiom, leading to an exponential-time algorithm. 

\myitem Rely on a \TUF$_{\{\wedge,\vee\}}$-solver through encoding all $store$ operations with $select$ and if-then-else expressions ($ite$). For example, 
$select(store(store(A,j_1,v_1),j_2,v_2),i)$ is rewritten into  $ite(i=j_2,v_2,ite(i=j_1,v_1,select(A,i)))$. The transformation introduces disjunctions, leading to an exponential-time  
 algorithm. 
\end{itemize}

\subsection{Fixed-size arrays and Constraint Programming}  \label{sec:array-fd}


A variant of $\TA$ can be dealt with in CP(FD): arrays are restricted to have a fixed and known size, while finite-domain constraints 
over indexes and elements are natively supported.

A logical fixed-size array is encoded \textit{explicitly} in CP(FD) solvers   by a fixed-size array (or sequence) of finite-domain variables. 
The {\it select} constraint is typically handled by 
  constraint {\sc element} $(i,A,v)$ \cite{HC88}.
The constraint holds iff $A[i]=v$, where $i$, $v$ are finite domain variables and $A$ is a fixed-size array.  
Local consistency filtering algorithms are  available for {\sc element}  at quadratic cost \cite{Bra01}.  
%
%
Filtering algorithms for $store$ constraints have been proposed in~\cite{CBG09}, with applications to software testing.   
The $store$ constraint can be reasoned about in CP(FD) by creating a new array of finite domain variables
and using filtering algorithms based on the content of arrays.  
Two such filtering algorithms for $select$ and $store$ are described in Section~\ref{sec:fdcc}, Figure~\ref{fig:cp-array}. 
%


Aside dedicated propagators, $store$ can also be totally removed \mycomment{from a CP(FD) framework}  through the introduction of 
reified case-splits (conditional constraints),  following the method of Section \ref{sec:array-symbolic}. 
Yet, this is notoriously inefficient here \mycomment{in CP(FD)}  because of the absence of adequate global filtering.      
%




{\it Terminology.} In this article, we consider  filtering over {\sc element} as implementing { local reasoning},  
while { global reasoning} refers to  deduction mechanisms working on  a global view of the constraint system, e.g.~taking into account all $select$/$store$.  
We will also use the generic terms {\sc Access} and {\sc Update} to refer to any correct filtering algorithm 
for $select$ and $store$ over finite domains.   
%
%













\section{Combining \CCA\  and \CPFD}  \label{sec:fdcc}

We present here our combination procedure for handling formulas over arrays and finite-domain indexes and elements.  
The resulting decision procedure natively supports finite-domain constraints and  combines global symbolic reasoning with local domain filtering. 
Moreover, we can reuse existing FD solvers in a black-box manner through a minimal API.

\subsection{Overview}

Our approach is based on combining symbolic global reasoning for arrays  and local filtering.   
The framework, sketched in Fig.~\ref{archi}, is built over three main ingredients:  

\begin{enumerate}

\item  local filtering  for arrays plus constraints on elements and indexes, named \CPFD,     

\item a lightweight  global symbolic reasoning procedure over arrays, named \CCA, 

\item a new   bi-directional communication mechanism between \CPFD\ and \CCA.    

\end{enumerate}

%


\noindent Let $\varphi$ be a conjunction of equalities, disequalities, array accesses ($select$) and updates ($store$), constraint on 
the size of arrays 
and other (arbitrary) constraints over elements and indexes. Our  procedure
takes $\varphi$ as input, and returns a verdict that can be either $sat$  or $unsat$. 
First, the formula $\varphi$ is preprocessed and dispatched between \CCA\  and \CPFD. 
More precisely, equalities and disequalities as well as array accesses and updates  go to both solvers. 
Constraints over elements and indexes go only to  \CPFD. 
%
%
The two solvers exchange the following information (Fig.~\ref{archi}):  
\CCA\  can communicate new equalities and disequalities among variables to \CPFD, as well as 
sets of variables being all different (i.e.,~cliques of disequalities);   
%
\CPFD\  can also communicate new equalities and disequalities to \CCA, based on domain analysis of variables. 
The communication mechanism and the decision procedures  are described more precisely in the rest of this section. 
%
%
\begin{figure}[htb]
\centering
\includegraphics[width=0.9\textwidth]{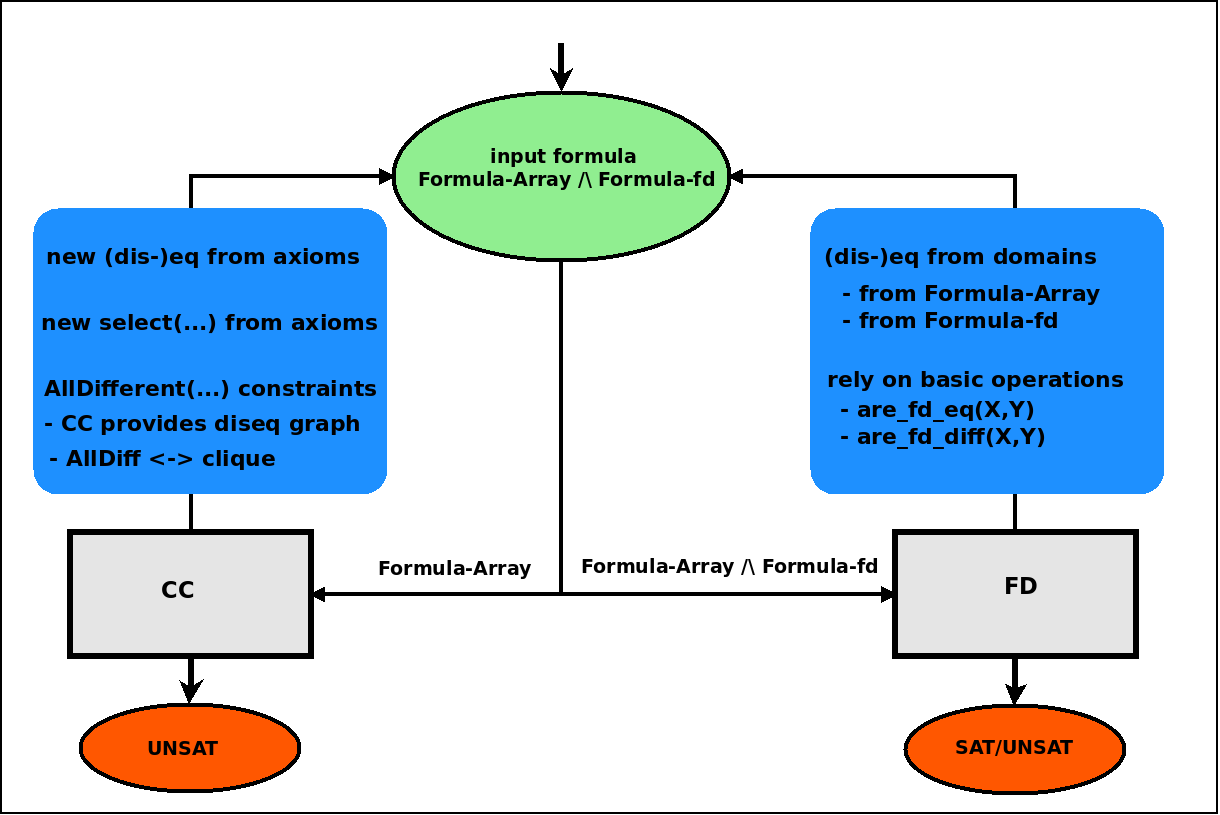}
	\caption{A bi-directional process for combining \CCA\  and \CPFD}\label{archi}
\end{figure}
%
%
%
%
%
%
\subsection{The \CCA\  decision procedure} \label{subsec:cc-arrays}

We  adapt 
the standard congruence closure algorithm 
into  a semi-decision procedure  \CCA\  for arrays. By semi-decision procedure, we mean  
that  deductions made by the procedure are correct w.r.t.~array axioms  
but may not be sufficient to conclude to $sat$ or $unsat$.  \CCA\  is correct (verdict can be trusted) but not complete 
(may output ``maybe'').  
For the sake of clarity, we refine the set of array axioms given in Section  \ref{sec:arrays} into 
an equivalent set of six operational rules (cf.~Figure~\ref{array:rules}), taking axioms and their contrapositives into account.

\begin{figure}[htbp]
\centering
\fbox{\small
\begin{tabular}{lcl}

\textit{(FC-1) \qquad }  & \ \bigstrut   &  $   i = j \longrightarrow select(A, i) = select(A, j)  $ \\

\textit{(FC-2) \qquad }   & \bigstrut  &  $ select(A, i) \neq select(A, j) \longrightarrow  i \neq j  $ \\

\textit{(RoW-1-1) \qquad }  & \bigstrut  &  $ i = j \longrightarrow select(store(A, i, e), j) = e  $ \\

\textit{(RoW-1-2) \qquad } & \bigstrut  &   $ select(store(A, i, e), j) \neq e  \longrightarrow   i \neq j  $ \\

\textit{(RoW-2-1) \qquad }  & \bigstrut  &  $ i \neq j \longrightarrow select(store(A, i, e), j) = select(A, j)  $ \\

\textit{(RoW-2-2) \qquad } & \bigstrut  &   $ select(store(A, i, e), j) \neq select(A, j)  \longrightarrow   i = j  $ \\

\end{tabular}
}
%
\caption{Rules for array axioms}\label{array:rules} 
\end{figure}

We adapt the congruence closure algorithm in order to handle these six rules.

\begin{itemize}

\myitem  Rules FC-1 and FC-2 are commonly handled with slight extension of congruence closure  \cite{NO80}, taking sub-terms into account.  
Each term  $t$ is now equipped with two sets $t.sup$ and $t.sub$ denoting  the sets of its  
direct super-terms and  sub-terms.      




\myitem  To cope with rules  RoW-1-1 to RoW-2-1,     we add a  mechanism of \emph{delayed evaluation}: 
for each term $t \mydef select(store(A,i,e),j)$, we put pairs   
$(i=j \triangleright t=e)$, $(t\neq e \triangleright t \neq j)$ and $(i \neq j \triangleright t = select(A,j))$ 
in a watch list. Whenever the left-hand side  of a pair in the watch list can be proved, we deduce that the corresponding right-hand side constraint holds.


\myitem For RoW-2-2, we rely  on  delayed evaluation, but only if term $select(A,j)$ is syntactically present in the formula. 

\end{itemize}


While implied disequalities are left implicit in standard congruence closure,  
we close  the set of disequalities (through FC-2 and RoW-1-2) in order to  benefit as much as possible from  rules RoW-2-1 and RoW-1-2.  
The whole procedure is described in Figure~\ref{algo:cc}. 
For the sake of conciseness, a few simplifications have been made: 
 we did not include  ranking optimisation of  congruence closure  (cf.~Section\ref{sec:congruence-closure});   
the unsatisfiability check {\tt check\_unsat()} is performed  
at the end of  main function \CCA\ while it could be performed on-the-fly when merging equivalence classes or adding a disequality;  
 the watch list  should be split into one list  of watched pairs per equivalence class,  allowing function {\tt check\_wl()} to iterate only over watched pairs  corresponding to  
modified equivalence classes.

This polynomial-time  procedure  is clearly not complete 
(recall that the satisfaction problem over arrays is NP-complete) 
but it implements a nice trade-off between standard congruence closure (no array axiom taken into account) 
and full closure at exponential cost (introduction of case-splits for RoW-* rules).

\begin{figure}[htbp]

\begin{center}
\fbox{\small 

\begin{minipage}{\textwidth}

{\bf global variable} wl := $\emptyset$; {\tt // watch list, elements of the form $(\psi \triangleright \varphi)$}

{\bf global variable} todo := $\emptyset$; {\tt // work list, elements of the form $\formula{t_1 = t_2}$ or $\formula{t_1 \neq t_2}$}

\smallskip

\begin{minipage}{\textwidth}
\begin{algorithm}[H]

{\bf function} \CCA $(\varphi)$: \tcp{$\varphi$ is an atomic constraint}

todo := $\{\varphi\}$\;

\While{$\text{todo} \neq \emptyset$}{

choose $\varphi' \in \text{todo}$ ; todo := todo - $\varphi'$ \;

update\_wl($\varphi$)\;

\Switch{$\varphi'$}{

\Case{$\formula{t_1 = t_2}$: }{
     $union(t_1,t_2)$ \; 
     close\_eq($find(t_1).super$)   \tcp*[l]{update variable todo (rule FC-1)}  
}

\Case{$\formula{t_1 \neq t_2}$: }{
     $t_1'$ := find($t_1$); $t_2'$ := find($t_2$);
     $t_1'.diff$ := $t_1'.diff + t_2'$;  
     $t_2'.diff$ := $t_2'.diff + t_1'$\; 
     close\_diff($t_1'$, $t_2'$)  \tcp*[l]{update variable todo (rule FC-2)}     

}

}

check\_wl() \tcp*[l]{update variables wl and todo (rules RoW-*)}

}

\lIf{\text{check\_unsat()}}{\Return UNSAT } \lElse{\Return OK} \;  

{\bf end}

\end{algorithm}
\end{minipage}

\begin{minipage}[t]{7.5cm}

\begin{algorithm}[H]
{\bf function} close\_eq($s$):\\
\tcp{elements in $s$ are pairs $(A,t)$ } 
\tcp{representing $ t \mydef select(A,j)$ } 
\tcp{for a given $j$ } 
 \ForAll{$(A,t),(A,t') \in s$}{ 
     todo := todo + $\formula{t=t'}$ \;     
  } 
\end{algorithm}

\begin{algorithm}[H]
{{\bf function}  $union(x,y)$: } \\




     $x'$ := $find(x)$; $y'$ := $find(y)$; $y'.parent$ := $x'$\; 
         $x'.diff$ := $x'.diff \cup y'.diff$ \; 
          $x'.super$ := $x'.super \cup y'.super$ \;
            $x'.sub$ := $x'.sub \cup y'.sub$ \;       

\end{algorithm}

\begin{algorithm}[H]
{\bf function} update\_wl($\varphi$):\\
\tcp{$Terms$ is the set of all terms} 
\tcp{seen so far} 
 
\ForAll{$t \in \varphi$ s.t.~ $t \mydef select(store(A,i,e),j)$}{ 

   wl := wl $\cup$ $\{( i=j \triangleright t=e )\}$\; 
   wl := wl $\cup$ $\{( t \neq e \triangleright i\neq j )\}$\;
   wl := wl $\cup$ $\{( i\neq j \triangleright t=select(A,j) )\}$\;

     \If{$t' \mydef select(A,j)$ $\in Terms$ }{   wl := wl $\cup$ $\{( t\neq t' \triangleright i=j )\}$}     
  }

 \ForAll{$t' \in \varphi$ s.t.~ $t' \mydef select(A,j)$}{ 

        \If{$t \mydef select(store(A,i,e),j)$ $\in Terms$ }{   wl := wl $\cup$ $\{( t\neq t' \triangleright i=j )\}$}   

}

\end{algorithm}

\begin{algorithm}[H]
{\bf function} check\_unsat():\\
\tcp{iterates over all terms seen so far,} 
\tcp{looking for contradiction}
 \ForAll{$t \in Terms$}{ 
     \lIf{diff($t$,$t$) }{\Return true}\;     
  } 
 \Return\ false; 
\end{algorithm}

{\bf function} equal(t,t'): \Return $find(t)$==$find(t')$;  

{\bf function} diff(t,t'): \Return $find(t) \in find(t').diff$;

\end{minipage}
\hfill 
\begin{minipage}[t]{8cm}

\begin{algorithm}[H]
{\bf function} close\_diff($t'_1$,$t'_2$):\\
\tcp{elements in $t'.sub$ are pairs $(t,A)$} 
\tcp{representing $t' = select(A,t)$ } 
 $s_1$ := $t'_1.sub$ ;   $s_2$ := $t'_2.sub$ \;  
 \ForAll{$(t_1,A),(t_2,A) \in s_1 \times s_2$}{ 
     todo := todo + $\formula{t_1 \neq t_2}$ \;     
  }  
\end{algorithm}

\begin{algorithm}[H]
{\bf function} $find(x)$:\\
 \If{$x.parent \neq x$}{$x.parent$ := $find(x.parent)$ } 
 \Return $x.parent$\;
\end{algorithm}

\begin{algorithm}[H]
{\bf function} check\_wl():\\

\ForAll{$p \mydef (\psi \triangleright \varphi) \in \text{wl}$}{ 

b:=false\; 

\Switch{partial\_eval($\psi$)}{

\Case{{\sc true}:}{
todo := todo + $\varphi$; b:=true\;   
}

\lCase{{\sc false}:}{
b:=true;
}

\lCase{{\sc unknown}:}{skip;}

}

\lIf{b}{wl := wl - $p$}\;   

}
\end{algorithm}

\begin{algorithm}[H]
 {\bf function} partial\_eval($\psi$):\\

 \Switch{$\psi$}{

 \Case{$\formula{t_1 = t_2}$: }{

\lIf{equal($t_1$,$t_2$)}{r :=  {\sc true}}

\lElseIf{diff($t_1$,$t_2$)}{r := {\sc false}}
   
\lElse r :=  {\sc unknown} \; 

\Return r\;
 }

 \Case{$\formula{t_1 \neq t_2}$: }{

\lIf{equal($t_1$,$t_2$)}{r := {\sc false}}

\lElseIf{diff($t_1$,$t_2$)}{r := {\sc true}}
   
\lElse r := {\sc unknown} \; 

\Return r\;
 }
 }
\end{algorithm}

\end{minipage}

\end{minipage}
}
\end{center}

\caption{The \CCA\ procedure} \label{algo:cc}
\end{figure}

\subsection{The \CPFD\  decision procedure}

We use existing propagators and domains for constraints over finite domains.  
Our approach requires at least array constraints for $select/store$ operations, and support 
of   {\sc Alldifferent} constraint  \cite{Reg94} is a plus.  
%
%
%
%
An overview of propagators for  {\sc Access} and {\sc Update} is provided in Figure~\ref{fig:cp-array}, where the propagators 
are written in a  simple pseudo-language.   {\tt I} and {\texttt E}  are variables, 
while {\tt A} and {\tt A'} are (finite) arrays of variables.  Each variable {\tt X} comes with a finite domain {\tt D(X)} (here a finite set). 
Set operations have their usual meaning, \texttt{X==Y} (resp.~{\tt X=!=Y}) makes variables \texttt{X} and \texttt{Y} equal (resp.~different),  
 \texttt{integer(X)?} is true iff \texttt{X} is instantiated to an integer value, 
and \texttt{success} indicates that the constraint is satisfied.

\begin{figure}[htbp]
\hrulefill \smallskip 

\begin{center}
\begin{minipage}{0.8\textwidth}

\small
\tt
Access(A,I,E) :  

 fixpoint( 
  
  \quad integer(I)?   A[I] == E, success 

  \quad ; 

  \quad  D(E) := D(E) $\cap$ $\bigcup_{i \in \text{D(I)}}$ D(A($i$))
   
  \quad ; 

  \quad  D(I) :=  $\{ i \in \text{D(I)} | \text{D(E)} \cap \text{D(A[i])} \neq \emptyset    \}$

)

\end{minipage}

\medskip 

\begin{minipage}{0.8\textwidth}

\small
\tt
Update(A,I,E,A') : 

 fixpoint( 
  
  \quad  integer(I)?   A'[I] == E, forall k $\neq$ I do   A'[k] == A[k], success

  \quad ; 

  \quad  D(E) := D(E) $\cap$ $\bigcup_{i \in \text{D(I)}}$ D(A'($i$))
   
  \quad ; 

  \quad  D(I) :=  $\{ i \in \text{D(I)} | \text{D(E)} \cap \text{D(A'[i])} \neq \emptyset    \}$

  \quad ; 

  \quad  forall  k $\not \in $ D(I) do    A'[k] == A[k] 

  \quad ; 

 \quad  forall  k $ \in $ D(I) do   D(A'[k])  :=  D(A'[k]) $\cap$  (D(A[k])$\cup$ D(E))

  \quad ; 

 \quad  forall k $\in$ D(I) do if (D(A[k]) $\cap$ D(A'[k]) $= \emptyset$) then I == k

)

\end{minipage}
\end{center}

\smallskip \hrulefill 

\caption{Standard implementations of constraints {\sc Access} and {\sc Update}} \label{fig:cp-array}
\end{figure}

\subsection{Cooperation between \CCA\  and \CPFD\ }

The cooperation mechanism involves both to know which kind of information can be exchanged, and how the two solvers synchronise together.  
Our main contribution here is twofold: we identify  interesting information to share, and we design 
 a method to tame the communication cost.  

\myparagraph{Communication from  \CCA\  to \CPFD.} 
Our implementation of \CCA\  maintains the set of disequalities and therefore both equalities and disequalities can be transmitted to \CPFD. 
Interestingly, disequalities can be communicated through {\sc Alldifferent}  constraints in order to increase the deduction capabilities
of \CPFD.  More precisely, any set of disequalities is captured by an undirected graph where each node is a term, and there is an edge between two
terms $t_1$ and $t_2$ if and only if $t_1 \neq t_2$. Finding cliques\footnote{A clique $C$ is a subset of the 
vertices such that every two vertices in $C$ are connected by an edge.} in the graph allows us to transmit {\sc Alldifferent} constraints
to \CPFD, e.g., $t_1 \neq t_2, t_2 \neq t_3, t_1 \neq t_3$ is communicated to \CPFD\ 
 using {\sc Alldifferent}$(t_1,t_2,t_3)$. 
These cliques can be sought dynamically during the execution of \CCA. 
Since finding  \textit{a largest clique} of a graph is NP-complete,  restrictions have to be considered. 
Practical choices are described in Sec.~\ref{sec:implem}. 

\myparagraph{Communication from  \CPFD\  to \CCA.}
\CPFD\  may discover new disequalities and equalities through filtering. For example, consider 
the  constraint $z \geq x \times y$ with domains $x \in 4..5$, $y \in 2..3$  and $z \in 8..12$. 
While no more filtering can be performed\footnote{Technically speaking, the constraint system is said to be {\it bound-consistent}.},  
we can still deduce that $x \neq y$, $x \neq z$ and $y \neq z$,  and transmit them  to \CCA. 
Note that, as \CCA has no special support for {\sc Alldifferent}, there is no need to transmit these inequalities under the form of this global
constraint in this case.
Yet, this information is left implicit in the constraint store of \CPFD\  and needs to be checked explicitly. 
But there is a quadratic number of pairs of variables, and (dis-)equalities could appear at each filtering step.  
Hence, the eager generation of all domain-based (dis-)equalities must be temperated in order to avoid a combinatorial
explosion.  We propose efficient ways of doing it hereafter. 


\myparagraph{Synchronisation mechanisms: how to tame communication costs.}
A purely asynchronous cooperation mechanism with systematic exchange of information between \CPFD\  and \CCA\  
(through suspended constraints and awakening over domain modification), 
 as exemplified 
in Fig.~\ref{archi}, appeared to be too expensive in practise. 
We managed  this problem through a reduction of the number of pairs of variables to consider (\textbf{critical pairs}, see after)  
and a  \textbf{communication policy} allowing tight control over expensive communications.  

\medskip 

\noindent {\bf 1.} The {\bf communication policy} obeys the following principles: 

\begin{itemize}
\myitem cheap communications are made in an asynchronous manner; 
\myitem expensive communications are made only on request, initiated by a \textit{supervisor};  
\myitem the two solvers run asynchronously, taking messages from the supervisor;  
\myitem  the supervisor is responsible to dispatch formulas to the solvers, to ensure a consistent view of the problem between \CPFD\  and \CCA, 
to forward answers of one solver to the other and  to send queries for expensive computations.  
\end{itemize}

\smallskip 

It turns out that all communications from \CCA\  to \CPFD\  are cheap, while 
 communications from \CPFD\  to \CCA\  are expensive. Hence, it is those communications which are made only upon request. 
Typically, it is up to the supervisor to explicitly ask if a given pair of variables is equal or different in \CPFD. 
Hence we have a total control on this mechanism. 
%

\medskip 

\noindent {\bf 2.} We also 
{\bf reduce the number of  pairs of variables to be checked} for (dis-)equality in \CPFD, by focusing only on pairs whose disequality  
will directly lead to new deductions in \CCA. For this purpose, we consider  pairs involved in the left-hand side of rules {FC-*}, {RoW-1-*} and {RoW-2-*}.   
Such pairs   will be   
 called \textit{critical}. 
Considering the six deduction rules of Section \ref{subsec:cc-arrays}, the  set of  critical pairs $\C$ of a formula $\varphi$ is defined as follows:   

\begin{itemize}
\myitem $\C_{FC}$ contains exactly all  pairs  $(select(A,i), select(A,j))$, where $A$, $i$ and $j$ appear syntactically in the formula (denoted $A,i,j \in \varphi$); 

\myitem   $\C_{RoW}$ contains exactly all   pairs $(i,j)$ and $(e,v)$  for each term $t \mydef select(store(A,i,e),j) \in \varphi$,  plus pairs  $(t,select(A,j))$ if  
$select(A,j) \in \varphi$. 


\myitem The set of critical pairs is defined by $\C\ \mydef \C_{FC} \uplus \C_{RoW} $.  

\end{itemize}
  
\smallskip 

The number of critical pairs $|\C|$ is still  quadratic, not in the number of variables but in the number of $select$. 
We choose to focus our attention only on the second class of critical pairs, namely $\C_{RoW}$:  they capture the specific essence of array axioms (besides FC)  
 and their number is only \textit{linear} in the number of $select$.  This restriction of critical pairs corresponds exactly to the pairs checked for equality or disequality 
in the {\sc WatchList} of the \CCA\ procedure (Section \ref{subsec:cc-arrays}).    
In practise, it appears that this reduction is manageable while still bringing interesting deductive power. 
A summary of the set of  pairs to be considered and their number is given in Table~\ref{tab:critical-pairs}.  

\begin{table}[htb]
\begin{center} 
\small
\begin{tabular}{|c|c|c|}

\hline

rules   $\bigstrut$  &  set of pairs    &      \# of pairs   \\

\hline

no restriction $\bigstrut$  &           $V \times V$                                           &    $O(|V|^2)$              \\

\hline 

FC-*, RoW-*       $\bigstrut$     &    $\C$   &    $O(|select|^2)$ \\ 


\hline 

FC-*  $\bigstrut$    &  $\C_{FC}$   &    $O(|select|^2)$ \\  

\hline

RoW-*  $\bigstrut$    &  $\C_{RoW}$    &    $O(|select|)$  \\   



\hline
\end{tabular}
\end{center} 
\caption{Number of pairs to consider for checking (dis-)equality in \CPFD} \label{tab:critical-pairs}
\end{table}

\myparagraph{The labelling procedure.} 
So far we have only considered propagation. However, while the propagated information is correct, it is not complete. Completeness 
is recovered through a standard labelling approach. 
We consider labelling in the form of $X=k \mbox{ or } X \neq k$. 
The labelling procedure constrains only \CPFD: it appears that flooding \CCA\  with all the new (dis)-equalities at each choice point 
was expensive and mostly worthless. In a sense, most labelling choices do not impact \CCA, and those which really matter are \textit{in fine} transmitted
  through queries about critical pairs.

\myparagraph{Complete architecture of the approach.} 
A detailed architecture of our approach can be found in Fig.~\ref{archi2}.  
Interestingly, \CCA\  and \CPFD\  do not behave in a  symmetric way:  \CCA\  transmits systematically to the supervisor 
all new deductions made and cannot be queried, while \CPFD\  transmits equalities and disequalities only upon request from the supervisor. 
Note also that \CCA\  can only provide a definitive $unsat$ answer (no view of non-array constraints) while \CPFD\  
can provide both  definitive $sat$ and $unsat$ answers. 
The list of critical pairs is dynamically modified by the supervisor:  pairs are added when new $select$ are deduced by \CCA\  
and already proved (dis-)equal pairs are removed. 
In our current implementation, the supervisor queries \CPFD\  on all active critical pairs at once. 
Querying takes place after each propagation step. 


\begin{figure}[htb]
 \centering
  \includegraphics[width = 0.9\textwidth]{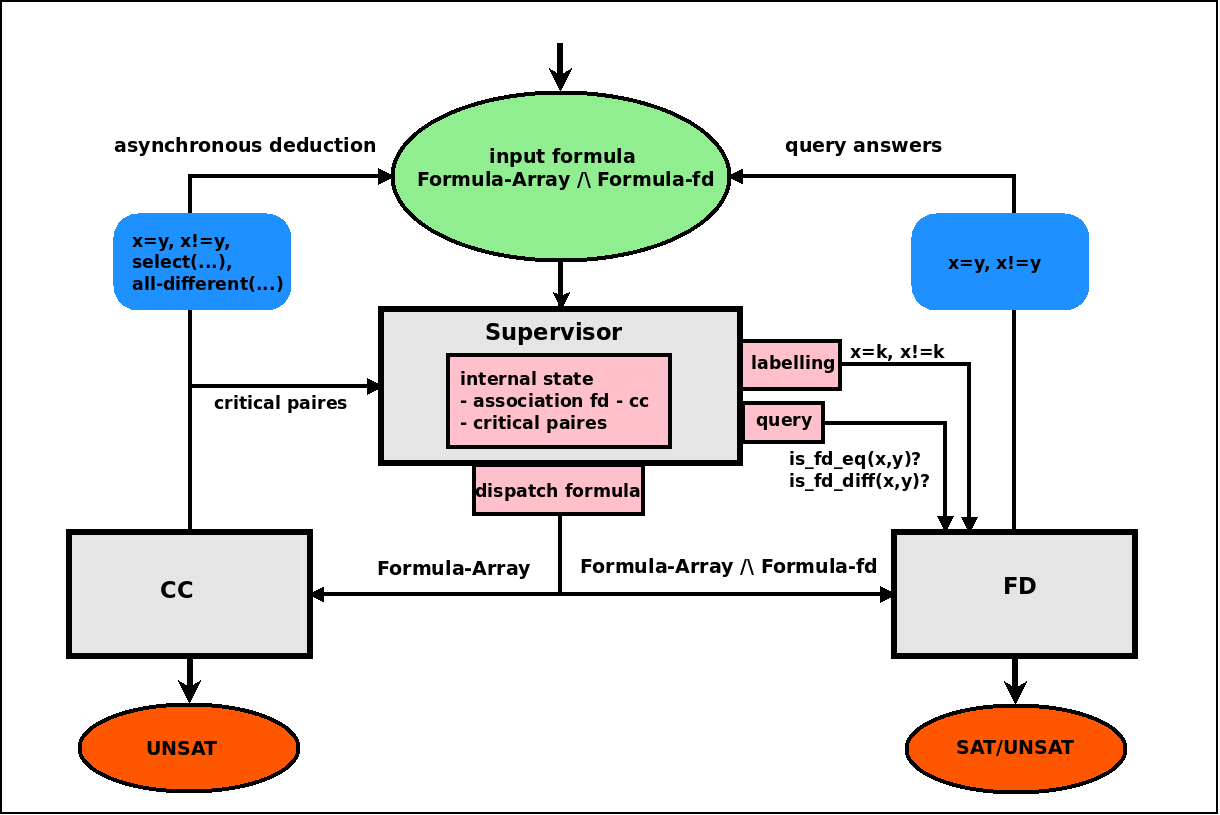}
\caption{Detailed view of the communication mechanism}\label{archi2}
\end{figure}

\myparagraph{API for the CP(FD) solver.} 
 While the approach requires  a dedicated implementation of the supervisor and \CCA\  (yet, most of \CCA\  is standard and easy to implement), 
any CP(FD) solver can be used as a black box, as long as it provides support for: 

\begin{itemize}

\myitem  the atomic constraints considered in the formula  ({\sc Access}, {\sc Update} and whatever constraints required over indexes and elements),   

\myitem     the two functions \texttt{is\_fd\_eq(x,y)} and \texttt{is\_fd\_diff(x,y)}, stating if two variables can be proved  equal or different.   
These two functions are either available or easy to implement in most CP(FD) systems. They are typically based on the available domain information, for example 
\texttt{is\_fd\_diff(x,y)} may return \textit{true} iff $D(\text{\tt x}) \cap D(\text{\tt y}) = \emptyset$. More precise (but more demanding) implementations can be used.  
For example, we can force an equality between {\tt x} and {\tt y} and observe  propagation. Upon  failure, we      
      deduce that {\tt x} and {\tt y} must be different.  

\end{itemize}



\myparagraph{Alternative design choices.} We discuss here a few alternative design solutions, and the reasons why we discarded them. 
We already pointed out that systematically transmitting  to \CCA\ all labelling choices  was inefficient (i.e.~we observed a dramatic drop in performance and no advantage in solving power), since most of these choices do not lead to relevant deduction in  \CCA.  
For the same reasons, it appears that transmitting to \CCA\  every instantiation obtained in \CPFD\ through propagation does not help. 
We also experimented an asynchronous communication mechanism for critical pairs. Typically, a dedicated propagator \texttt{critical-pair(X,Y)} 
was launched each time \CCA\ found a new critical pair. The propagator awakes on modifications of $D(\text{\tt X})$ or $D(\text{\tt Y})$, and checks 
 if any of \texttt{is\_fd\_eq(x,y)} or   \texttt{is\_fd\_diff(x,y)} is conclusive. If yes, the propagator sends the corresponding relations to \CCA\ 
and successfully terminates. Again, this alternative design appears to be  inefficient, the \texttt{critical-pair} propagators being continuously 
awoken for no real benefit.

\subsection{Properties of the framework}

\myparagraph{Comparing \fdcc\ with standard approaches for arrays.} 
Table \ref{tab:comparison} gives a brief comparison of \fdcc, \CCA\ and \CPFD. 
Compared to a standard CP(FD) approach, the main advantage of \fdcc\ is to add 
a symbolic and global  deduction  mechanism. Yet the approach is still limited to fixed-size arrays. 
Compared to a standard symbolic approach for  $\TA$, \fdcc\ enables to reason natively about finite domains variables  
and contains FD constraints over both array elements and indexes. However, \fdcc\ cannot deal with unknown-size arrays 
and cannot be easily integrated into a Nelson-Oppen combination framework.

\begin{table}[htbp]
\begin{center}  \small

\begin{tabular}{|l|c|c|c|}

\hline
         &   \CPFD  &  \CCA   & \fdcc   \\   
\hline

add FD constraints          &    \mycheckmark      &   \mybadmark      &     \mycheckmark       \\ 

\hline

add SMT constraints          &    \mybadmark         &    \mycheckmark       &     \mybadmark        \\ 

\hline

reasoning over domains          &     \mycheckmark       &     \mybadmark       &     \mycheckmark       \\ 

\hline



global symbolic deduction          &     \mybadmark        &     \mycheckmark      &    \mycheckmark        \\ 



\hline

unknown-size arrays          &    \mybadmark         &    \mycheckmark       &    \mybadmark         \\ 

\hline
\end{tabular}
\end{center}
\caption{Comparison between \fdcc, \CPFD\ and \CCA} \label{tab:comparison}
\end{table}


\myparagraph{Theoretical properties of the framework.} Let  $\varphi$ be a conjunctive formula over arrays and finite-domains variables and constraints. 
  A \CPFD\ propagator is  {\it correct}  if every  
filtered value does not belong to any solution of $\varphi$. Moreover, a correct \CPFD\ propagator is {\it strongly correct} if it correctly evaluates  
fully-instantiated instances of the problem (i.e.~the propagator distinguishes between solutions and non-solutions).   
We denote by \fdcc-propagation and \CPFD-propagation the propagation steps of \CPFD\ and \fdcc. 
\CPFD-propagation is limited to domain filtering, while \fdcc-propagation considers (dis-)equalities propagation as well.  
A decision procedure is said to be correct if both positive and negative results can be trusted, 
and complete if it terminates.










\begin{theorem} 
Assuming that \CPFD\ filtering is strongly correct, the following properties hold:  
%
%
(i) \fdcc-propagation terminates,  
%
(ii) \fdcc-propagation is correct, 
%
and (iii) \fdcc\ is correct and complete.  


\end{theorem}

\begin{proof}
 {\it Proof.}  (i) \CPFD\ and \CCA\ can only send a bounded amount of information from one to each other:    \CPFD\ can send to \CCA\ a number of new (dis-)equalities in $O(|\varphi|^2)$ (critical pairs),  
and \CCA\ can send to \CPFD\ a number of  new (dis-)equalities in $O(|store| + |select|^2)$.  
Since each solver alone terminates, the whole \fdcc-propagation step terminates. 
(ii) Correctness of \fdcc-propagation comes directly from the correctness of the \CCA\ procedure (easily derived by comparing the deduction rules and the axioms of $\TA$) and the assumed correctness of \CPFD-propagation. 
(iii) The labelling procedure ensures termination since the number of variables does not change along the resolution process (\CCA\ can deduce new terms, 
but no new variables). Negative results (UNSAT) can be trusted because \fdcc-propagation is correct, while positive results (SAT) can be trusted because \CPFD-propagation is strongly correct.  
Altogether,  we deduce that \fdcc\ is correct and complete.\qed
\end{proof} 


\subsection{Running examples}
 Consider the array formulas extracted from Fig.~\ref{ex1}. \CPFD\ solves each formula in less than $1$sec. 
For Prog1,  \CCA\  immediately determines that (1) is $unsat$, as $i=j$ allows to merge $e$ and $f$, which are declared to be different. 
For Prog2,  
in \CCA, the formula is not detected
as being $unsat$ (the size constraint over $A$ being not taken into account), but rule \textit{(FC-2)} produces  the new disequalities  
$i \neq j$,  $i \neq k$ and  $j \neq k$. 
Then, the  
two cliques  $(e,f,g)$ and $(i,j,k)$ are identified.  In \CPFD, the domains of $i,j,k$ are pruned to $0..1$ and 
\textit{local} filtering alone cannot go further. However, when considering the 
cliques previously identified, two supplementary
\textit{global constraints}  are added to the constraint store: {\sc Alldifferent}$(e,f,g)$ and {\sc Alldifferent}$(i,j,k)$. 
The latter and the pruned
domains of $i,j,k$ allow  \fdcc\   to conclude that (2) is $unsat$. This example shows that it is worth supporting {\sc Alldifferent}. 










\section{Implementation and experimental results}   \label{sec:implem-experiments}

In order to evaluate the potential interest of the proposed approach, we developed a prototype constraint solver that
combines both the \CCA\  and \CPFD\  procedures. The solver was then used to check the satisfiability of
large sets of randomly generated formulas and structured formulas. This section describes our tool called {\sc fdcc}, and 
details our experimental results.

\subsection{Implementation of \fdcc}  \label{sec:implem}

We developed \fdcc\ as a 
constraint solver over $\TA$ augmented with finite domains arithmetic.
It takes as input  formulas written in the above theory
and classifies them as being $sat$ or $unsat$. In the former case, the tool also returns a 
solution (i.e., a model) under the form of a complete instantiation
of the variables. Formulas may include array select and store, array size declaration, variable equalities and disequalities,  
finite domains specifications  and (both linear and non-linear) arithmetic constraints on finite domain variables. 

\fdcc\  is implemented on top of SICStus Prolog and is about 1.7 KLOC. 
It exploits the {\tt clpfd} library \cite{COC97} which provides an optimised implementation
of {\sc Alldifferent} as well as efficient filtering algorithms for arithmetical constraints over finite domains. 
The FD solver is extended with our own implementations of the array 
select and store operations~\cite{CBG09}. 
%
%
Communication is implemented through message passing and awakenings. 
{\sc Alldifferent}  constraints are added each time a $3$-clique is detected. 
Restricting clique computations to $3$-cliques is advantageous to master the combinatorial explosion of
a more general clique detection. Of course, more interesting deductions may be missed (e.g.,~$4$-cliques) but
we hypothesise that these cases are seldom in practise.
The $3$-clique detection is
launched each time a new disequality constraint is considered in \CCA. 
 CPU runtime  is measured  on 
an Intel Pentium 2.16GHZ machine running Windows XP with 2.0GB of RAM.

\subsection{Experimental evaluation on random instances}

Using randomly generated
formulas is advantageous for evaluating the approach, as there is no bias in the choice of problems.
However, there is also a threat to validity as random formulas might not fairly represent
reality. In SAT-solving, it is well known that solvers that
perform well on randomly generated formulas are not necessary good on real-world problems. To mitigate the risk,
we built a dedicated random generator that produces realistic instances.

\myparagraph{Formula generation.}
We distinguish \textit{four different classes of formulas}, depending on whether 
linear arithmetic constraints are present or not (in addition to array constraints) and 
whether array constraints are (a priori) ``easy'' or ``hard''. 
Easy array constraints are built upon three arrays, two without any $store$ constraint, 
and the third created by two successive stores. 
Hard array constraints  are built upon $6$ different arrays involving long chains of store 
(up to $8$ successive stores to define an array). 
The four classes are: 
\begin{itemize}
\item AEUF-I (easy array constraints),  
\item AEUF-II (hard array constraints),  
\item AEUF+LIA-I (easy array constraints plus linear arithmetic), 
\item AEUF+LIA-II (hard array constraints plus linear arithmetic). 
\end{itemize}

We performed two distinct experiments: in the first one we try to 
\textit{balance sat and unsat formulas} and \textit{more or less complex-to-solve formulas} 
by varying the formula length, around and above the \textit{complexity threshold}, while
in the second experiment, we regularly increase the formula length in order to cross the \textit{complexity threshold}.
Typically, in both experiments, small-size random formulas are often easy to prove $sat$ 
and large-size random formulas are often easy to prove $unsat$.  
In our examples,  
formula length varies from $10$ to $60$. 
%
In addition, the following other parameters are set up:  formulas contain around $40$ variables (besides arrays), arrays have size $20$ 
and all variables and arrays range over domain \mycomment{$0..50$}$0..1000$, so that enumeration alone is unlikely to be sufficient. 



\myparagraph{Properties to evaluate.}
We are interested in the following two aspects when comparing the solvers:  
(1) the ability to solve as many formulas as possible, 
and (2) the average computation time on easy formulas.  
%
%
These two properties are equally important in verification settings: solving 
a high ratio of formulas is of primary importance, but 
a solver able to solve many formulas with an important overhead may be less interesting than a faster solver 
missing only a few difficult-to-solve formulas.    

\myparagraph{Competitors.}
We submitted the formulas to three versions of \fdcc. The first version is the standard \fdcc\ described so far. 
 The second version
includes only the \CCA\  algorithm  while the third version implements only the \CPFD\  approach.  
%
In addition, we also use two witnesses, \textsc{hybrid} and \textsc{best}.   
\textsc{hybrid} represents a naive concurrent (black-box) combination of \CCA\  and \CPFD: both solvers  run in parallel, 
the first one getting an answer stops the other. 
\textsc{best}  simulates a portfolio procedure with ``perfect'' selection heuristics:  for each formula, we simply take 
the best result among \CCA\  and \CPFD. 
\textsc{best} and \textsc{hybrid} are not implemented, but deduced from 
results of \CCA\  and \CPFD. \textsc{best}  serves as a reference point, representing the best possible  black-box combination, while \textsc{hybrid} 
serves  as witness, in order to understand if \fdcc\ goes further in practise than just a naive black-box combination. 
All versions are correct and complete, allowing a fair comparison. The \CCA\  version requires that the  labelling procedure communicates 
 each (dis-)equality choice to  \CCA\   in order to ensure correctness.


\myparagraph{Results of the first experiment.}
 For each formula, a time-out of $60$s was positioned. 
We report the number of  $sat$, $unsat$ and $timeout$ answers for each solver in 
Tab.~\ref{tab:r3}.

\begin{table}[htbp]

\small

\begin{minipage}{0.4\textwidth}
\begin{center}
\begin{tabular}{|c||c|c|c|c|}


\hline

          &     \multicolumn{4}{c|}{All categories}  \\

          &    \multicolumn{4}{c|}{(369 formulas)}  \\

\hline 
           &   S & U & TO &   T  \\

\hline

\CCA\            & 29  & 115  & 225  & 13545    \\

\hline

\CPFD\               & 154  & 151  & 64  &  3995   \\

\hline

\fdcc\                   & \bf 181  & 175  & \bf 13  &  \bf 957  \\

\hline

\textsc{best}        & 154  & 175  & 40  & 2492    \\

\hline

\textsc{hybrid}            & 154   & 175  & 40  &  2609   \\

\hline 
\end{tabular}
\end{center}
\end{minipage}
\hfill 
\begin{minipage}{0.6\textwidth}

\begin{center}
\begin{tabular}{|c||c|c|c|c||c|c|c|c|}


\hline

          & \multicolumn{4}{c||}{AEUF-I}  & \multicolumn{4}{c|}{AEUF-II}   \\

          & \multicolumn{4}{c||}{(79)}  & \multicolumn{4}{c|}{(90)}     \\

\hline 
          & S & U & TO & T                  &     S & U & TO & T          \\

\hline

\CCA\         & 26  & 37  & 16  & 987            & 2  & 30  & 58  & 3485                        \\

\hline

\CPFD\        & 39  & 26  & 14  & 875           & 35 & 18  & 37 &  2299                \\

\hline

\fdcc\         & \bf 40  & 37  & \bf 2 &  \bf 144      & \bf 51 & 30 & \bf 9 &  \bf 635         \\

\hline

\textsc{best}        & 39  & 37  & 3  & 202            & 35  & 30  & 25  & 1529                   \\

\hline

\textsc{hybrid}        & 39  & 37  & 3  & 242            & 35  & 30  & 25  & 1561                    \\

\hline 
\hline

          &  \multicolumn{4}{c||}{AEUF+LIA-I} &    \multicolumn{4}{c|}{AEUF+LIA-II}    \\

          &  \multicolumn{4}{c||}{(100)} &    \multicolumn{4}{c|}{(100)}      \\

\hline 
                   &  S & U & TO & T &      S & U & TO &   T  \\

\hline

\CCA\                           & 1  & 21 & 78  & 4689             & 0  & 27  & 73  &  4384        \\

\hline

\CPFD\                          & 50  & 47  & 3  & 199              &  30  & 60  & 10  & 622                   \\

\hline

\fdcc\        &      \bf 52  & 48  & \bf 0  & \bf 24         & \bf 38         & 60  & \bf 2  & \bf 154             \\

\hline

\textsc{best}                                   & 50   & 48  & 2  & 139             
        & 30  & 60  & 10  & 622                \\

\hline

\textsc{hybrid}                           & 50   &  48 & 2  & 159            
              & 30  & 60  & 10  & 647                  \\

\hline 
\end{tabular}
\end{center}

\end{minipage}

\medskip

S : \# sat answer, U : \# unsat answer, TO : \# time-out (60 sec), T: time in sec.
\caption{Experimental results of the first experiment} \label{tab:r3}

\end{table}

As expected for pure array formulas (AEUF-*), \CPFD\  is better on the $sat$ instances, 
and \CCA\  behaves in an opposite  way. Performance of \CCA\  decreases quickly on hard-to-solve $sat$ formulas.   
Surprisingly, the two procedures behave quite differently in presence of arithmetic constraints: 
we observe that $unsat$ formulas become often easily  provable \textit{with domain arguments}, explaining 
why \CPFD\  performs better and \CCA\  worst  compared to the AEUF-* case. 
Note that computation times reported in Tab.~\ref{tab:r3}  are dominated by the number of time-outs (TO), since here solvers often quickly succeed  
or fail. Hence \textsc{best} and \textsc{hybrid} do not show any significant difference in computation time, 
while  in case of success,  \textsc{best} is systematically 2x-faster than \textsc{hybrid}. 
 Results show  that: 
\begin{itemize}
\item 
\fdcc\  \textit{solves strictly more formulas} than \CPFD\  or \CCA\  taken in isolation, 
and even more formula than \textsc{best}. Especially, there are $22$ formulas solved only by \fdcc, 
and \fdcc\ shows 5x-less TO than \CPFD\  and   3x-less TO than \textsc{best}. 


\item 
\fdcc\ yields only \textit{a very affordable overhead} over \CCA\  and \CPFD\  when they succeed. 
 \fdcc\  is at worst 4x-slower than \CCA, \CPFD\ and \textsc{best} when they succeed. On average it is  1.5x-slower (resp.~1.1x-slower) than  \CCA\ and \CPFD\ (resp.~\textsc{best}) when they succeed.   
%
%

\item 
These results hold for the four classes of programs, for both $sat$ and $unsat$ instances, 
and a priori easy or hard instances. Hence, \fdcc\  \textit{is much more robust} than \CPFD\  or \CCA. 

\end{itemize}

\myparagraph{Results of the second experiment.}
In this experiment, $100$ formulas of class AEUF-II are generated with length $l$, $l$ varying from $10$ to $60$.
While crossing the {\it complexity threshold}, we record the number of time-outs (TO, positioned at $60$sec).  
In addition, we used two metrics to evaluate the
capabilities of \fdcc\ to solve formulas,  {\bf Gain} and {\bf Miracle}, defined as follows:
\begin{itemize}
\item {\bf Gain}:
each time \fdcc\  classifies a formula that none of (resp.~only one of) \CCA\  and \CPFD\  can classify, {\bf Gain} is rewarded by $2$ (resp.~$1$); 
%
each time \fdcc\  cannot classify a formula that one of (resp.~both) \CCA\  and \CPFD\  can classify, {\bf Gain} is penalised by $1$ (resp.~$2$). 
Note that the $-2$ case never happened during our experiments. 
%

\item {\bf Miracle} is  the number of times \fdcc\ gives a result  when both \CCA\  and \CPFD\  fail. 
\end{itemize}

\medskip 

Fig.~\ref{res} shows the number of solved formulas for each solver, the number of formulas which
remain unsolved because of time-out, and both the values of {\bf Gain} and {\bf Miracle}.
\begin{figure}[t]
  \begin{center}
  \includegraphics[scale=0.60]{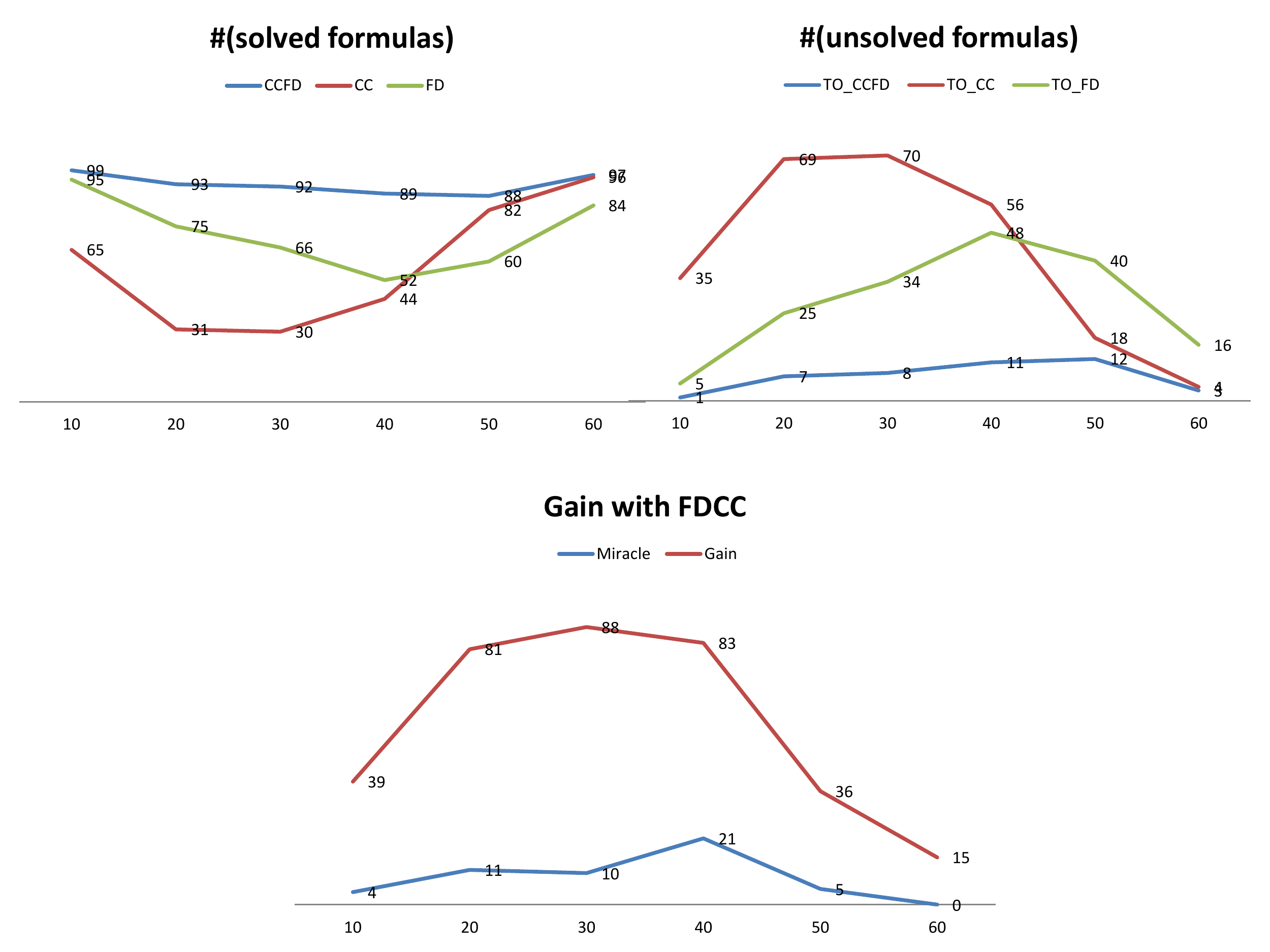}
		\caption{Experimental results for the $2^{nd}$ experiment}\label{res}
  \end{center}
\end{figure}
We see that the number of solved formulas is always greater for \fdcc\ (about 20\% more than \CPFD\  and about 70\%
more than \CCA). Moreover,  \fdcc\ presents maximal benefit for formula length   
in between $20$ and $40$, i.e.~for a length close to the {\it complexity threshold}, 
meaning that the relative performance is better on hard-to-solve formulas. 
  For these lengths,  the number of unsolved formulas is always less  than $11$ with \fdcc,  
 while it is always greater than $25$ with both \CCA\  and \CPFD.  


\myparagraph{Conclusion.} Experimental results show that \fdcc\  performs better than \CPFD\  and \CCA\  taken in isolation, 
especially on hard-to-solve formulas, and is very competitive with portfolio approaches mixing \CPFD\  and \CCA. 
More precisely,  

\begin{itemize}

\myitem 
\fdcc\ solves strictly more formulas than its competitors (3x-less TO than \textsc{best})  
and shows a low overhead over its competitors (1.1x-average ratio when \textsc{best} succeeds). 
%

\myitem  
relative performance is better on hard-to-solve formulas
than on easy-to-solve formulas, suggesting that it becomes especially worthwhile to combine 
global symbolic reasoning  with local filtering  when hard instances have to be solved.

\myitem \fdcc\ is both reliable and robust on the class of considered formulas ($sat$ or $unsat$, {\it easy-to-solve} or {\it hard-to-solve}).

\end{itemize}

 This is particularly interesting in verification  settings, since it means that \fdcc\ is clearly 
preferable to the standard \CPFD-handling of arrays in any context, i.e.,~whether we want to solve a few complex formulas 
or to solve as many as formula in a small amount of time.






\section{Extensions of the core technique}  \label{sec:extensions}

In this section, we discuss several extensions of \fdcc. We focus on extensions 
of $\TA$ relevant to software verification. Interestingly, the combination framework can be reused without any modification, 
only the \CCA\ or \CPFD\ solvers must be extended.

\subsection{Uniform arrays} 


%

Many programming languages offer the developer to initialise arrays with the same constant value, typically $0$, or the same general expression. Dealing efficiently with constant-value initialisation is necessary in any concrete implementation of a software verification framework.   
%
%
%
%
%
%
In order to capture this specific data structure, we add at the formula level an array term of the form $K_{<e>}$, where $e$ represents a term.  
For these arrays, called {\it uniform arrays}, we introduce the following extra rule: $\forall i,  select(K_{<e>},i) = e$.

Uniform arrays can  be handled in \fdcc\ as follows: 
%
%
(i) add a new   rule in \CCA\    rewriting $select(K_{<e>},\_)$ into $e$,  
%
(ii) in \CPFD,  either unfold each array $K_{<e>}$ and fill it with variables equal to  $e$, or (preferably) add a special kind of ``folded'' array  such that   
{\sc Access} always returns $e$  and  {\sc Update} creates an unfolded version filled with $e$ terms. 


\subsection{Array extensionality}


Software verification over array programs sometimes involves (dis-)equalities over whole arrays. For example, programs  that perform  
string comparison often include string-level primitives.   
%
%
%
%
%
%
%
 For this purpose, formulas can be extended with  equality and disequality predicates over arrays, denoted $=_A$ and $\neq_A$ in the 
{\it extensional theory of arrays} \cite{BB08}. 

Array equality can be directly  handled by  congruence-closure on array names in  \CCA\  and by index-wise unification of arrays in \CPFD.   
  When checking satisfiability of quantifier-free formulas,  any array disequality $A_1 \neq_A A_2$ can be    replaced  
 by a standard disequality  $ select(A_1,x) \neq select(A_2,x) $, where $x$ is a fresh variable. 
This preprocessing-based solution is sufficient for both  \CCA\ and \CPFD. Yet,  implementing 
a dedicated constraint for array disequality  can lead to better propagation.  
Such a constraint is described in Figure~\ref{fig:cp-array-extension}. 
\mycomment{, where {\tt A} and {\tt A'} are arrays and {\tt I} is any index witnessing the disequality.}








\begin{figure}[htbp]
\hrulefill \smallskip 

 \begin{center}









\begin{minipage}{0.95\textwidth}

\small
\tt
Diff-array(A,I,A') :- 

 fixpoint(

\quad  integer(I)?    A$[i]$ =!= A'$[i]$, success 

\quad ; 

%

\quad  D(I) :=  D(I) $\backslash$ $\{ k \  | \  \text{\tt A}[k] = \text{\tt A'}[k]    \}$    

)

\end{minipage}

\end{center}

\smallskip \hrulefill 

\caption{CP(FD) constraint for array disequality } \label{fig:cp-array-extension}
\end{figure}

We provide a small example illustrating the advantage of the {\tt Diff-array} constraint over introducing a fresh variable $x$ such that $ select(A_1,x) \neq select(A_2,x) $. 
Let us consider two arrays $A_1$ and $A_2$ with constant size $N$.  Moreover, let us assume that for all $i \in 1..N$, $A_1[i]=A_2[i]=i$. Constraint {\tt Diff-array($A_1$,$x$,$A_2$)} immediately 
returns $unsat$ since $D(x)$ is reduced to $\emptyset$ by the second rule. On the other hand,  {\sc Access} constraint for $select$ propagates  $ D(select(A_1,x)) = D(select(A_2,x)) = [1..N]$. 
From this point, no  more propagation is feasible through the $\neq$ constraint, especially $D(x)$ is not reduced at all. In that case, $unsat$ can be proved only after enumerating the whole domain of $x$ 
 ($N$ values).

\subsection{Arrays with non-fixed (but bounded) size} \label{sec:non-fixed-arrays}

 We have assumed so far that arrays have a known fixed  size. However, practical software verification  
also involves arrays of unknown size, for example in unit-level verification.  We propose the following scheme for extending our approach to arrays with non-fixed (but bounded) size.   
%
%
%
%
%
%
 Formulas are extended with a new function $size: A \mapsto \N$, and every $select$ or $store$ index over an array $A$ is constrained to be less or equal to $size(A)$. 
Moreover, we assume that each term $size(A)$  has a known upper-bound.  


 This extension does not modify the \CCA\ part of the framework, since $\TA$ already considers unbounded arrays.   
On the other hand, the filtering algorithms associated to constraints over arrays  must be significantly revised. We take inspiration from previous work of one of the authors \cite{CBG09}, describing 
 an {\sc Update} constraint for memory heaps whose sizes are a priori unknown.   
In this work, memory heaps can be either {\it closed} or {\it unclosed}. We  adapt this notion to arrays:  
closing an array comes down to fixing its size to a constant. As a result, the filtering algorithm is parametrised with the state of the array and deductions may be different whether the array is closed or unclosed. The closed case reduces to standard array filtering (Figure \ref{fig:cp-array}).  
%
The unclosed case is significantly different: unclosed arrays have  a non-fixed size  and only part of their elements are explicitly represented. They can be encoded internally 
by partial maps from integers to logical variables.   
Filtering is rather limited in that case, but as soon as the array gets closed, more deductions can be reached.


We present a simple  implementation of constraints over unclosed arrays in Figure~\ref{fig:cp-array-unknown-size}, finer propagation can be 
derived from ideas developed in \cite{CBG09}. 
Propagators for {\sc Access-unclosed} and {\sc Update-unclosed} mostly look like  their counterparts over closed arrays. Note the use of operations \texttt{?D(A[k])} and \texttt{merge(A,k,X)} \--  where {\tt A} is an array, {\tt k} $\in \N$ and {\tt X} a logical variable \-- instead of  \texttt{D(A[k])} and \texttt{A[k] == X } in the case of closed arrays. These two new operations account for the case where no pair {\tt (k,Y)} is recorded so far in {\tt A}. 
 In that case, \texttt{?D(A[k])} returns $\top$ (the whole set of possible values for elements) and  \texttt{merge(A,k,X)} adds the pair {\tt (k,X)} to 
the set of explicitly described elements of {\tt A}. We suppose we are given a function \texttt{is-def(A,k)} to test if index {\tt  k} and its corresponding element are explicitly stored in {\tt A}.  Finally,  the \texttt{fill} operation ensures that all pairs of an array recognised as \texttt{closed} will be explicitly represented.  
 \begin{figure}[htbp]
\hrulefill \smallskip 

\begin{center}

\begin{minipage}{0.95\textwidth}

\small
\tt
Access-unclosed(A,I,E) : 

 fixpoint( 

    \quad  closed(A)?  Access(A,I,E), success

 \quad ; 

  \quad  integer(I)?  merge(A,i,E), success

  \quad ; 

  \quad  D(E) := D(E) $\cap$ $\bigcup_{i \in \text{D(I)}}$ ?D(A($i$)) 
   
  \quad ; 

  \quad  D(I) :=  $\{ i \in \text{D(I)} | \text{D(E)} \cap \text{?D(A[i])} \neq \emptyset    \}$


)

\end{minipage}





























 





\medskip

-------------- 

\medskip 

\begin{minipage}{0.95\textwidth}

\small
\tt

closed(A): integer($S_A$)?  fill(A), success

\medskip

fill(A):  forall $\text{\tt i} \leq S_A  \text{ s.t. }\neg \text{is-def(A,i)}$ do: merge(A,i,N$_i$), with N$_i$  fresh 


\medskip

?D(A[k]):  if  is-def(A,k) then D(A[k])  else $\top$

\medskip

merge(A,k,E):   if is-def(A,k) then A[k] == E  else  A := A[k $\leftarrow$ E]

\end{minipage}

 \medskip

 -------------- 

 \medskip 

\begin{minipage}{0.95\textwidth}

\small
\tt
Update-unclosed(A,I,E,A') :

  fixpoint( 
  
  \quad  closed(A) and closed(A')?    Update(A,I,E,A'), success    

  \quad ; 

  \quad  closed(A) or closed(A')?  $S_A$ == $S_{A'}$    

  \quad ; 

  \quad  integer(I)?   merge(A',I,E)  

  \quad ; 

  \quad  D(E) := D(E) $\cap$ $\bigcup_{i \in \text{D(I)}}$ ?D(A'($i$))
   
  \quad ; 

  \quad  D(I) :=  $\{ i \in \text{D(I)} | \text{D(E)} \cap \text{?D(A'[i])} \neq \emptyset    \}$ 

  \quad ; 

  \quad  forall  k $\in $ [1 .. max($S_A$)] $\backslash$ D(I) do:  

 \quad  \quad  \quad  if is-def(A,k)  then merge(A',k,A[k]), 

 \quad  \quad  \quad     if is-def(A',k) then merge(A,k,A'[k])   

  \quad ; 

 \quad  forall  k $ \in $ D(I) s.t. is-def(A',k)   do:    

 \quad  \quad  \quad      D(A'[k])  :=   D(A'[k]) $\cap$  (?D(A[k])$\cup$ D(E))   

  \quad ; 

 \quad  forall k $\in$ D(I) do:  if (?D(A[k]) $\cap$  ?D(A'[k]) $= \emptyset$) then I == i   

)

\end{minipage}

\end{center}

\smallskip \hrulefill 

 \caption{Implementation of CP(FD) Constraints for arrays of unknown size } \label{fig:cp-array-unknown-size}
 \end{figure}



%

\subsection{Maps}

Maps extend arrays in two crucial  ways: indexes (``keys'') do not have to be integers,  and they can be both added and removed. 
General indexes open the door to constraints over hashmaps, which are useful in  many application areas, 
while removable indexes are essential to model memory-heaps with dynamic (de-)allocation \cite{Bornat-00,CBG09}.  
%


Maps  come with the $select$, $store$ and $size$ functions, plus functions $delete: H \times I \mapsto H$ (remove a key and its associated entry from the map)  
and   $keys:  H \times I \mapsto \B $, true iff  index $i$ is mapped in  $H$ (we sometimes denote $keys$ as a predicate).  The semantics is given by the set of   axioms given in Figure~\ref{maps:axioms},    
inspired from \cite[Chap.~11]{BM-07}~\footnote{We add the KoW-2 and KoD-2 axioms that are missing in the first edition of the book. The authors acknowledge the error on the book's website.}. 


\begin{figure}[htbp]

\centering

\fbox{\small 
\begin{tabular}{lcl}

\textit{(FC) \qquad }  & \ \bigstrut   &  $   i = j \longrightarrow select(H, i) = select(H, j)  $ \\

\textit{(RoW-1) \qquad }  & \bigstrut  &  $ i = j \longrightarrow select(store(H, i, e), j) = e  $ \\

\textit{(RoW-2') \qquad }  & \bigstrut  &  $ i \neq j \wedge keys(H,j)  \longrightarrow select(store(H, i, e), j) = select(H, j)  $ \\

\textit{(RoD-1) \qquad }  & \bigstrut  &  $ i \neq j \wedge keys(H,j) \longrightarrow select(remove(H, i), j) = select(H, j)  $ \\


\textit{(KoW-1) \qquad }  & \bigstrut  &  $ i = j \longrightarrow keys(store(H, i, e), j)  $ \\

\textit{(KoW-2) \qquad }  & \bigstrut  &  $ i \neq j \longrightarrow keys(store(H, i, e), j) = keys(H, j)  $ \\

\textit{(KoD-1) \qquad }  & \bigstrut  &  $ i = j \longrightarrow \neg keys(delete(H, i), j)  $ \\

\textit{(KoD-2) \qquad }  & \bigstrut  &  $ i \neq j \longrightarrow keys(delete(H, i), j) = keys(H, j)  $ \\

\end{tabular} 
}
\caption{Axioms for the theory of maps }\label{maps:axioms}
\end{figure}




 Interestingly, maps without size constraints can be encoded into pure arrays~\cite{BM-07}  using  
  two arrays  $A_K:  I \mapsto \B$ and $A_E:  I \mapsto E$  for each map $H:  I \mapsto E$. 
Array $A_K$ models the fact that a key is mapped in $H$ (value $\mytrue$) or not (value $\myfalse$),  array $A_E$ represents the relationship  between mapped keys and their associated values in $H$.  
The encoding works as follows: 

\begin{itemize}
\myitem   $select(H,j) = v $ becomes     $select(A_E,j) = v \wedge select(A_K,j)=\mytrue$,    
\myitem   $H'=store(H,i,v)$ becomes  $A'_E = store(A_E,i,v) \wedge A'_K = store(A_K,i,\mytrue)$, 
\myitem   $H'=delete(H,i)$  becomes  $A'_E = A_E  \wedge A'_K = store(A_K,i,\myfalse)$, 
\myitem   $keys(H,i)$ becomes   $select(A_K,i) = \mytrue$,
\myitem   $\neg keys(H,i)$ becomes   $select(A_K,i) = \myfalse$.  
\end{itemize}




For the \CPFD\ part, Charreteur \textit{et al.}~\cite{CBG09} provides dedicated  propagators   in the flavor of those presented in Section~\ref{sec:non-fixed-arrays}.   
There is yet a  noticeable difference with the case of non-fixed size arrays: the absence of    relationship between the size of a map (i.e.,~its number of mapped keys) and the value of its 
indexes.    
 It implies for example that map closeness  is not enforced through labelling on the size, but directly   through labelling on the ``closeness status'', either 
setting it to \textit{true} (no more unknown elements in the map) or keeping it to {\it false} but adding  a fresh variable to a yet unmapped index value.  
%
%

\section{Related work}  \label{sec:related}

This paper is an extension of a preliminary version presented at CPAIOR 2012 \cite{BG-12}. It contains detailed descriptions and explanations on the core technology, formulated in complete revisions of Sections~\ref{sec:background} to~\ref{sec:fdcc}. It also presents new developments and extensions in a completely new Section~\ref{sec:extensions}. 
Moreover, as it discusses adaptations of the approach for  several  extensions of the theory of arrays relevant to software verification, it also contains a  deeper and updated  description  of related work (Section \ref{sec:related}).

\myparagraph{Alternative approaches to FDCC.} 
We sketch three  alternative methods for handling array constraints over finite domains, 
and we argue why we do not choose them.  
First, one could  think of embedding a CP(FD) solver in a SMT solver, as one theory solver among others, the array constraints being 
handled by a dedicated solver. 
As already stated in introduction,  standard 
cooperation framework like Nelson-Oppen (NO) \cite{NO79} require that supported theories have an infinite model,   
which is not the case for Finite Domains.   

Second, one could simply use a simple  concurrent black-box combination (first solver to succeed wins). 
Our greybox combination scheme is more complex (yet still rather simple), but performance is much higher as demonstrated 
by our experiments. Moreover, we are still able to easily reuse existing CP(FD) engines thanks to a small easy-to-provide API.

Third, one could encode all finite-domain constraints into boolean constraints and use a SMT solver equipped with a decision procedure for the standard theory of arrays.  
Doing so, we give away the possibility of taking advantage of  the high-level structure of the initial formula.   
Recent works on finite but hard-to-reason-about constraints, such as floating-point arithmetic \cite{BCG13}, 
modular arithmetic \cite{GLM10} or bitvectors \cite{BHP10}, suggests that it can be much more efficient in some cases to keep the high-level view of the formula.

\myparagraph{Deductive methods and SMT frameworks.} 
It is well known in the SMT community that 
solving formulas over arrays and  integer arithmetic  in an efficient way through  
NO is difficult. 
Indeed, handling \textit{non-convex theories}  in a correct way  
 requires to propagate all \textit{implied disjunctions of equalities},  
which  may be much more expensive than satisfiability checking \cite{BCFGS-09}. 
Delayed theory combination \cite{BBCJRRS-05,BCFGS-09} requires only  the propagation of implied equalities, at the price  
 of  adding  new boolean variables for all potential equalities between variables.        
%
%
%
Model-based theory combination~\cite{MB08}  
aims at mitigating this potential overhead  through lazy propagation of equalities.  
%
%
%

Besides,  $\TA$ is  hard to solve by itself.  
Standard symbolic approaches have already been sketched in Section~\ref{sec:array-symbolic}.  
%
The most efficient approaches combine preprocessing for removing as many RoW terms as possible with ``delayed'' inlining of array axioms for the remaining 
RoW terms.   New lemmas corresponding roughly to critical pairs can be  added  on-demand  to the DPLL top-level~\cite{BNOT-06},  
%
or they can   be     incrementally discovered   through an abstraction-refinement scheme  \cite{BB08}. 
Additional performance can be obtained through frugal ($\approx$ minimal) instantiation of array axioms 
\cite{GKF-08}.

\myparagraph{Filtering-based methods.}
Consistency-based filtering approaches for array constraints are already discussed in Section~\ref{sec:array-fd}. 
A logical combination of {\sc Element} constraints (with disjunctions) can express {\sc Update} constraints. However, a dedicated  {\sc Update} constraint, billed as a global constraint, implements more global reasoning and is definitely more efficient in case of non-constant indexes.  
%
%
%
The work of Beldiceanu et al. \cite{BCD05} has shown that it is
possible to capture global state of several {\sc Element} constraints with a finite-state automaton. This approach could be followed as well to capture {\sc Update} constraint, but we do not foresee its usage for implementing global reasoning over a chain of {\sc Access} and {\sc Update}. Indeed, this would require the design of a complex automaton dedicated to each problem.
Based on a \CCA\ algorithm, our approach captures a global state of a set of {\sc Access} and  {\sc Update} constraints but it is also only symbolic and thus less effective than using dedicated constraints. In our framework,
the \CCA\  algorithm cannot prune the domain of index or indexed variables. In fact, our proposition has more similarities with the proposition of 
Nieuwenhuis on his DPLL({\sc Alldifferent}) framework\footnote{\tt http://www.lsi.upc.edu/~roberto/papers/CP2010slides.pdf}, where 
the idea is to benefit from the efficiency of several global constraints in the DPLL algorithm for SAT encoded problems. 
In \fdcc, we derive {\sc Alldifferent} global constraints from the congruence closure algorithm for similar reasons. Nevertheless, 
our combined approach is fully automated, which is a keypoint to address array constraint systems coming from various software
verification problems. 



\myparagraph{Combination of propagators in CP.} 
Several possibilities can be considered to implement constraint propagation when multiple propagators are available  \cite{SS08}.  
%
First, an external solver can be embedded as a new global constraint in \CPFD, as done for example on the {\sc Quad} global constraint for continuous domains \cite{LMR05}. 
 This approach 
offers global reasoning over the constraint store. However, 
 it requires fine control over the awakening mechanism of the  new global constraint.  
%
%
%
A second approach consists in calling both solvers in a concurrent way. Each of them is launched on a distinct thread, and both
      threads prune a common constraint store that serves of blackboard. This approach has been successfully implemented in 
      Oz \cite{VBD03}. The difficulty is to identify which information must be shared, and to do it efficiently. 
%
 A third approach consists in building  a master-slave combination process where one of the solvers (here \CCA) drives the computation and call the other (\CPFD). 
 The difficulty here  is to understand when the master must call the slave.   
 We follow mainly the second approach, however a third agent (the supervisor) acts as a lightweight master over \CCA\  and \CPFD\  
to synchronise both solvers through queries.


\section{Conclusions and perspectives}  \label{sec:conclusion}
  
 This article describes an approach for  
solving  conjunctive quantifier-free formulas combining arrays and  finite-domain  constraints over indexes and elements. 
We sketch an original decision procedure that combines ideas from symbolic reasoning    
and finite-domain constraint solving for array formulas.
%
 The communication mechanism proposed in the article
lies on the opportunity of improving the deductive capabilities of the congruence closure algorithm with finite domains information.  
We also propose ways of keeping  communication overhead tractable. According to our knowledge, this is the first time such a combination framework at the interface of 
CP and SMT is proposed and implemented into a concrete prototype.
Experiments show that our  approach  performs better than any portfolio  combination of  a symbolic solver and a filtering-based solver. 
Especially, our procedure enhances greatly the deductive power of standard CP(FD) approaches for arrays. 
%
%
%
Future works include integrating \fdcc\ into an existing  software verification tool (e.g., \cite{BH11,Gotlieb-09}) in order to improve its efficiency over programs with arrays.   

\ACKNOWLEDGMENT{%
We are very grateful to Pei-Yu Li who proposed a preliminary encoding of \fdcc\ during her trainee period, 
and Nadjib Lazaar for comparative experiments with OPL.
}

%
%
%




\end{document}